\documentclass{article}

\usepackage[utf8]{inputenc}
\usepackage[table]{xcolor}

\usepackage{hyperref}
\definecolor{ToCgreen}{RGB}{0, 128, 0}
\hypersetup{
    colorlinks=true,
    linkcolor=blue,
    urlcolor=cyan,
    citecolor=ToCgreen,
}
\usepackage{amsmath,amsthm,amssymb,bm,caption,subcaption,nicefrac}
\usepackage[nameinlink]{cleveref}

\usepackage{multirow}
\usepackage{arydshln}
\usepackage{csquotes}
\usepackage{anyfontsize}
\usepackage[shortlabels]{enumitem}
\usepackage{natbib}
\usepackage{booktabs}
\usepackage{wrapfig}
\usepackage{tikz}
\usepackage{pgfplots}
\usepackage{authblk}

\pgfplotsset{compat=1.18}

\usepackage{algorithm}
\usepackage[noend]{algorithmic}

\newtheorem{theorem}{Theorem}

\newtheorem{lemma}[theorem]{Lemma}

\newtheorem{definition}[theorem]{Definition}

\usepackage{fullpage}

\newcommand{\Ber}{\mathrm{Ber}}

\newcommand{\cM}{\mathcal{M}}
\newcommand{\cA}{\mathcal{A}}

\renewcommand{\epsilon}{\varepsilon}
\newcommand{\eps}{\epsilon}
\newcommand{\N}{\mathbb{N}}
\newcommand{\Z}{\mathbb{Z}}
\newcommand{\R}{\mathbb{R}}

\newcommand{\E}{\mathbb{E}}
\newcommand{\cD}{\mathcal{D}}
\newcommand{\cP}{\mathcal{P}}

\renewcommand{\phi}{\varphi}

\newcommand{\cO}{\mathcal{O}}
\newcommand{\cQ}{\mathcal{Q}}
\newcommand{\cE}{\mathcal{E}}

\newcommand{\teps}{\tilde{\eps}}

\newcommand{\Geo}[1]{\mathrm{Geom}_{#1}}
\newcommand{\Exp}[1]{\mathrm{Exp}_{#1}}

\DeclareMathOperator{\supp}{supp}

\renewcommand{\setminus}{\smallsetminus}

\usepackage[disable]{todonotes}
\newcommand{\pasin}[1]{\todo[color=blue!40]{{\it Pasin:~}#1}}

\newcommand{\sasha}[1]{\todo[color=cyan]{{\it Sasha:~}#1}}
\newcommand{\pritish}[1]{\todo[color=blue!20]{{\it Pritish:~}#1}}

\newcommand{\mechanism}{\mathcal{M}}
\newcommand{\dataset}{D}
\newcommand{\datasetSet}{\mathcal{D}}
\newcommand{\resultSet}{\mathcal{O}}

\newcommand{\adversary}{\cA}

\newcommand{\IT}[4]{\mathrm{IT}^{#1}(\mathcal{F}_{{#3}, {#2}};{#4})}

\newcommand{\renyidiv}[3]{\mathrm{D}_{#1}\left({#2}~\|~{#3}\right)}

\newcommand{\algorithmicparameters}{\textbf{Parameters:}}
\newcommand{\PARAMETERS}{\item[\algorithmicparameters]}
\newcommand{\algorithmicpparameters}{\phantom{\textbf{Parameters:}}}
\newcommand{\PPARAMETERS}{\item[\algorithmicpparameters]}

\newcommand{\ans}[2]{#1 {\tiny $\pm$ #2}}

\allowdisplaybreaks

\title{Private Hyperparameter Tuning with Ex-Post Guarantee}

\author{
Badih Ghazi, Pritish Kamath, Alexander Knop \\
Ravi Kumar, Pasin Manurangsi, Chiyuan Zhang
}
\affil{Google Research\thanks{Authors are in alphabetical order of last name.}}

\begin{document}
    
    \maketitle

    \begin{abstract}
        The conventional approach in differential privacy (DP) literature formulates the privacy-utility tradeoff with a ``privacy-first'' perspective: 
        for a predetermined level of privacy, a certain utility is achievable. 
        However, practitioners often operate under a ``utility-first'' paradigm, prioritizing a desired level of utility and then determining 
        the corresponding privacy cost.
        
        \citet{WRLWN19} initiated a formal study of this ``utility-first'' perspective by introducing ex-post DP. 
        They demonstrated that by adding correlated Laplace noise and progressively reducing it on demand, a sequence of increasingly
        accurate estimates of a private parameter can be generated, with the privacy cost attributed only to the least noisy iterate released. 
        This led to a Laplace mechanism variant that achieves a specified utility with minimal privacy loss.
        However, their work, and similar findings by \citet{WRWR22}, are primarily limited to simple mechanisms based
        on Laplace or Gaussian noise. 
        
        In this paper, we significantly generalize these results.  In particular, we extend the work of \citet{WRLWN19} and \citet{LT19}
        to support any sequence of private estimators, incurring at most a doubling of the original privacy budget. 
        Furthermore, we demonstrate that hyperparameter tuning for these estimators, including the selection of an optimal privacy budget, 
        can be performed without additional privacy cost. 
        Finally, we extend our results to  ex-post R\'{e}nyi DP, further broadening the applicability of utility-first privacy mechanisms.
    \end{abstract}

    \section{Introduction}

    Many applications of machine learning and statistics involve computation on sensitive data, necessitating privacy-preserving techniques. 
    In recent years, differential privacy (DP)~\citep{DMNS16} has become one of the most rigorous formalization of privacy, with many practical 
    applications~\citep{AbadiCGMMT016,yu2021differentially,mehta2023towards,tang2024private,census2020, israelNationalRegistry,WZLDSG20}.
    Recall that an algorithm is DP if the output distributions on two neighboring inputs are close, where the closeness is determined by the 
    \emph{privacy budget}\footnote{%
        Our work also applies to \emph{approximate-DP} with $\delta$ parameter; see \Cref{sec:prelim} for the definition.
    } $\eps$:
    \begin{definition}[Pure Differentially Privacy,~\citep{DMNS16}] \label{def:dp-intro-exante}
        For $\eps \geq 0$, a mechanism $\mechanism$ with input from $\datasetSet$ and output from $\resultSet$ is  
        \emph{$\eps$-differentially private} (or simply, \emph{$\eps$-DP)} iff $\Pr[\mechanism(D) = o] \le e^\eps \Pr[\mechanism(D') = o],$
        for all $o \in \resultSet$ and neighboring datasets $\dataset, \dataset' \in \datasetSet$.
        \footnote{We also assume for simplicity that $\resultSet$ is finite; it is simple to extend the results to the infinite case.}
    \end{definition}

    For reasons that will become clear soon, we refer to the classic DP definition above as \emph{ex-ante} DP.
    
    \paragraph{Utility-First DP Mechanisms.} One of the main challenges in deploying DP is to ensure that the output remains useful. 
    In particular, real-world deployments are often constrained by utility requirements. 
    For example, in ML training, one may wish to ensure that the model accuracy meets a certain threshold. 
    Similarly, in statistical applications, one may wish to guarantee that the \emph{relative} error of the estimated population is small
    (e.g., \citep{GhaziKKMPSWW22}). 
    Such desiderata may not be compatible with the ex-ante DP (\Cref{def:dp-intro-exante}) since it \emph{a priori} specifies a fixed privacy budget $\eps$.
    This motivated \cite{WRLWN19} to propose the notion of \emph{ex-post} DP, where the privacy budget $\eps$ can \emph{depend on the output} of the mechanism,
    as formalized below.
    \begin{definition}[Ex-post (Pure-) DP, \citep{WRLWN19}]
    \label{def:dp-intro-expost}
        For a function $\teps : \resultSet \to \R^{\geq 0}$,  
        a mechanism $\mechanism$ with input from $\datasetSet$ and output from $\resultSet$ is
        \emph{ex-post $\teps$-DP} iff $\Pr[\mechanism(D) = o] \le e^{\teps(o)} \Pr[\mechanism(D') = o],$
        for all $o \in \resultSet$ and neighboring datasets $\dataset, \dataset' \in \datasetSet$.
    \end{definition}

    Observe that any ex-post $\teps$-DP mechanism is ex-ante $\eps$-DP where $\eps = \max_{o \in \cO} \teps(o)$. 
    Furthermore, ex-post DP can also be used as a \emph{privacy filter} to guarantee ex-ante DP~\citep{LebensoldPB24}. 
    Roughly speaking, given a total privacy budget $\eps$ for ex-ante DP, we run multiple ex-post algorithms where we subtract the realized privacy 
    budget $\teps(o)$ from $\eps$ until the latter is exhausted.

    Thus, the main question in ex-post DP becomes: What is the smallest privacy budget needed to produce an output that passes the desired utility bar? 
    As pointed out in \citep{WRLWN19}, a simple algorithm here is the ``doubling'' method where we start from a small privacy budget, run the ``base'' 
    (ex-ante DP) algorithm with this budget, and continue until we find an acceptable output\footnote{Checking whether an output passes a utility bar must also be done with DP.}. 
    While simple, this doubling method can result in the privacy budget as large as four times\footnote{A factor of two due to having to apply the composition theorem 
    to sums up all the budget, and another factor of two from the potential misalignment between the doubling exponential grid and the optimal budget.} the optimal budget. 
    Although there has been no improvement to this for general mechanisms, \cite{WRLWN19} and later \cite{WRWR22}, building on an earlier work by \cite{KoufogiannisHP16}, 
    gave an elegant improvement for the simple Laplace and Gaussian mechanisms that allows for a finer control of privacy budget increment than doubling and also just pays
    for the final privacy budget, instead of the total privacy budget (via composition). 
    Alas, their method does not apply to more complex mechanisms, such as the seminal DP-SGD algorithm~\citep{AbadiCGMMT016} that is ubiquitous in private ML applications.

    \paragraph{Hyperparameter Tuning with DP.} A related challenge in private ML deployments is hyperparameter tuning. 
    A naive solution here is to run any standard hyperparameter tuning algorithm and compute the total budget via  composition theorems.
    However, this results in a prohibitive blow-up in the privacy budget, depending on the number of times the base algorithm is invoked. 
    \cite{LT19} devised a simple algorithm but with a surprising guarantee. Their algorithm performs hyperparameter tuning on any ex-ante $\eps$-DP by 
    running it possibly multiple times (based on a carefully chosen distribution) and outputting the best found parameter.
    Even though the algorithm may be run many times, they show that the privacy budget incurred is only $3\eps$.
    Furthermore, they show that, any ``weakly useful'' ex-ante DP hyperparameter tuning algorithm must incur privacy budget at least (roughly) $2\eps$. 
    A follow-up work by \cite{PT22} closed this gap by giving an algorithm with privacy budget arbitrarily close to $2\eps$, and further generalized this 
    to work with R\'enyi DP~\citep{Mir17}. 
    Although the task of optimizing the privacy budget in ex-post DP framework seems similar to hyperparameter tuning where $\eps$ is a parameter, none of 
    the aforementioned works \citep{LT19,PT22} applies to this setting since they require the base mechanism to have a fixed value of $\eps$ in the ex-ante DP framework.

    \subsection{Our Contributions}

    In this work, we present the first hyperparameter tuning algorithm with ex-post DP guarantees. 
    Our algorithm can take in multiple base mechanisms $\mechanism_1, \dots, \mechanism_d$ where $\mechanism_i$ is ex-ante $\eps_i$-DP. 
    It then runs these mechanisms (possibly multiple times, based on carefully crafted distributions) and select the ``best'' output.
    The ex-post DP guarantee is that, if the output comes from the base mechanism $\mechanism_i$, then the privacy budget spent is only (roughly) $2\eps_i$.
    We consider this \emph{counterintuitive} and highly \emph{surprising} given that other base mechanisms $\mechanism_j$ with higher budget (i.e., $\eps_j > \eps_i$) 
    might be run en route and their output considered as part of the selection, nevertheless, our algorithm does not have to pay for this higher privacy budget $\eps_j$! 
    We are unaware of a similar phenomenon in DP.  

    While our algorithm (which works for multiple mechanisms with different $\eps_i$'s) is a significant generalization of those of \cite{LT19,PT22} 
    (which only work for a single mechanism in the ex-ante setting), our privacy analysis is arguably simpler than theirs. 
    In particular, the proof of our main privacy theorem (\Cref{theorem:max-score}) draws inspiration from that of the Sparse Vector Technique~\citep{DworkNRRV09} and 
    is elementary. 
    We hope that the resulting simplicity will help further elucidate the underlining principles behind DP hyperparameter tuning.
    We note that our hyperparameter tuning is well suited for the task of optimizing the privacy budget given the privacy bar in ex-ante DP, since we can set the
    ``score'' in the selection step to be based on the privacy budget and whether the privacy bar is passed. 

    Finally, we introduce a notion of ex-post R\'enyi DP and showed that hyperparameter tuning with R\'enyi DP is also possible (\Cref{theorem:max-score-rdp});
    in addition, we prove a connection between ex-post R\'enyi DP and ex-post approximate-DP and constructed a privacy filter that allows composing together
    a sequence of ex-post R\'enyi DP mechanism into a ex-ante R\'enyi DP guarantee (which could allow using this algorithm in practical systems that want to 
    provide ex-ante guarantees).
    Our technique, which applies to any mechanism including the aforementioned DP-SGD, is far more general than those in \citep{WRLWN19,WRWR22}, 
    which only applies to Laplace or Gaussian mechanisms. 
    To demonstrate this, we provide experiments that empirically show that our algorithm outperforms those in \citep{WRLWN19,WRWR22} for linear regression (using the conversion from
    ex-post R\'enyi DP to ex-post approximate-DP).
        
    \section{Preliminaries}
    \label{sec:prelim}
    
    Let $\datasetSet$ be a set of datasets. We write $D \sim D'$ as a shorthand for a pair of neighboring input datasets (in $\datasetSet$). 
    Let $\resultSet$ be any set; for  simplicity, we assume that $\resultSet$ is discrete.  
    We say that a function $\mechanism$ mapping $\dataset \in \datasetSet$ to a distribution over $\resultSet$ is a \emph{mechanism} with
    input from $\datasetSet$ and output from $\resultSet$.
        
    \paragraph{Ex-Ante DP.}
    While we have defined (ex-ante) \emph{pure}-DP in \Cref{def:dp-intro-exante}, it will be useful to recall other variants of DP. 
    We start with approximate-DP, which allows an additional additive error $\delta$ in the difference in the two probabilities, as defined below.
    When $\delta = 0$, this coincides with \Cref{def:dp-intro-exante}. 
    \begin{definition}[Differential Privacy,~\citep{DMNS16}]
        For $\eps, \delta \geq 0$, a mechanism $\mechanism$ with input from $\datasetSet$ and output from $\resultSet$ is  
        \emph{ex-ante $(\eps, \delta)$-differentially private} (or simply, \emph{$(\eps, \delta)$-DP)} iff $\Pr[\mechanism(D) \in S] \le e^\eps \Pr[\mechanism(D') \in S] + \delta$,
        for all $S \subseteq \resultSet$ and all $\dataset \sim \dataset'$. 
    \end{definition}

    Modern private learning is largely based on DP-SGD~\citep{AbadiCGMMT016}.
    The privacy analysis of such a mechanism, which involves both subsampling and composition, is often done through R\'{e}nyi DP~\cite{Mir17}, which we recall here.

    Let $\alpha > 1$, and $P$ and $Q$ be two distributions on $\resultSet$.  
    Let $\renyidiv{\alpha}{P}{Q}$ denote the \emph{R\'enyi divergence} of $P$ from $Q$, 
    i.e., $\renyidiv{\alpha}{P}{Q} = \frac{1}{\alpha - 1}\log \sum_{o \in \resultSet} (P(o))^\alpha (Q(o))^{1 - \alpha}$.

    \begin{definition}[R\'enyi DP, \citep{Mir17}]
        For $\alpha > 1$, $\eps \geq 0$, 
        a mechanism $\mechanism$ with input from $\datasetSet$ and output from $\resultSet$
        is \emph{ex-ante $(\alpha, \eps)$-R\'{e}nyi DP} 
        (or simply, \emph{$(\alpha, \eps)$-RDP)} iff
        $\renyidiv{\alpha}{\mechanism(\dataset)}{\mechanism(\dataset')} \le \eps$
        for all $\dataset \sim \dataset'$.
    \end{definition}

    \paragraph{Ex-Post DP.} We also need approximate and R\'{e}nyi variants of ex-post DP. We start with the former since
    it is defined in the literature before this paper.
    \begin{definition}[Ex-post DP, \cite{WRLWN19}]
        For a function $\teps : \resultSet \to \R^{\geq 0}$  and $\delta > 0$,
        a mechanism $\mechanism$ with input from $\datasetSet$ and output from $\resultSet$ is
        \emph{ex-post $(\teps, \delta)$-DP} iff for all $S \subseteq \resultSet$ and all $\dataset \sim \dataset'$, 
        \[
                \sum_{o \in S} \Pr[\mechanism(\dataset) = o] \leq
                \sum_{o \in S} e^{\teps(o)} \Pr[\mechanism(\dataset') = o] + \delta.
        \]
    \end{definition}
    Again, when $\delta = 0$, this coincides with \Cref{def:dp-intro-expost}. 
    Next, we introduce ex-post R\'{e}nyi DP.

    \begin{definition}[Ex-post RDP] 
        For a function $\eps : \resultSet \to \R^{\geq 0}$ and  $\alpha > 1$,          
        a mechanism $\mechanism$ with input from $\datasetSet$ and output from $\resultSet$ is \emph{ex-post 
        $(\alpha, \eps)$-RDP} iff for all $D \sim D'$,
        \[
            \sum_{o \in \resultSet} \frac{
                (\Pr[\mechanism(D) = o])^{\alpha}
            }{
                (e^{\eps(o)} \cdot \Pr[\mechanism(\dataset') = o])^{\alpha - 1}
            } \le 1.
        \]
    \end{definition}
    We note that if $\teps$ is a constant function, ex-ante and ex-post are equivalent for all DP notions stated.

    In the case of ex-ante DP, it is known that ex-ante pure-DP is a stronger notion than ex-ante RDP, 
    which in turn is stronger than ex-ante approximate-DP~\citep{Mir17}. We can show here that a similar result holds for their ex-post variants, as stated below. 
    The proof, which follows its ex-ante counterpart, is deferred to the Supplementary Material.
    \begin{lemma} \label{lem:rdp-v-dp}
        Let $\teps : \resultSet \to \R^{\geq 0}$ be a function and $\alpha > 1$, 
        $\delta \in [0, 1]$ be constants.
        \begin{itemize}[nosep]
            \item If $\mechanism$ is ex-post $\teps$-DP, then $\mechanism$ is 
                ex-post $(\alpha, \teps)$-RDP.
            \item If $\mechanism$ is ex-post $(\alpha, \teps)$-RDP, then $\mechanism$ is 
                ex-post $\left(\teps', \delta\right)$-DP, 
                where $\teps'(o) = \teps(o) +  \frac{\log 1 / \delta}{\alpha - 1}.$
        \end{itemize}
    \end{lemma}

    \paragraph{DP Selection Problem.} The main focus of our paper is on the \emph{DP selection} problem, which can be defined as follows. 
    There are $d$ mechanisms $\cM_1, \dots, \cM_d: \cD \to \cO$ where $\cM_i$ is ex-ante DP (or ex-ante RDP). 
    Following \citep{PT22}, we assume that $\cO$ is a totally ordered set.
    The goal is to, after running $\cM_1, \dots, \cM_d$ possibly multiple times, output $(o, i)$ where $i \in [d]$ and $o \in \cO$ is an output from one of the runs of $\cM_i$. Occasionally, we also allow an output $\perp$ to indicate that no good output was found.

    A classic application of DP selection is in \emph{DP hyperparameter tuning} of ML mechanisms.
    Here, each $\cM_i$ can represent the mechanism with different configuration of parameters (including the privacy budgets) and the output $\cM_i(D)$ is the private ML
    model together with the accuracy score on the test set.\footnote{%
        If the test set is considered sensitive, then we can add noise to achieve DP with respect to the test set.
    }
    Our setting also generalizes the widely-used exponential mechanism in which case $\cM_i$ can be thought of as 
    outputting the DP score of the $i$th candidate~\citep{McSherryT07}.
    \sasha{Should we say `outputting the noised out score of the $i$th candidate` instead?}
    
    \paragraph{Probability Notation.}
    For a distribution $\cP$, let $\supp(\cP)$ denote its support. 
    For $i \in \supp(\cP)$, let ${\cP}(i)$ denote the probability mass (resp.,  density) at $i$.
    For a subset $I$, let ${\cP}(I) = \sum_{i \in I} {\cP}(i)$ (resp.,  $\int_{I} \cP(i) di$). 
    Let $\Ber(p)$ denote the Bernoulli distribution with parameter $p$, i.e., the distribution on $\{0, 1\}$
    such that the probability of $1$ equals $p$.  
    Let $\Geo{p}$ denote the geometric distribution with \emph{failure probability} $p \in [0, 1]$,
    i.e., the distribution on $\Z_{\geq 0}$ such that $\Geo{p}(k) = (1 - p)p^k$.  
    Throughout this work, we will use the following property of the Geometric distribution in our proofs:
    \begin{align} \label{eq:coupling-geom}
    \Geo{p}(u) \leq p^{u - v} \cdot \Geo{p}(v) & & \forall u, v \in \Z \text{ such that } u \leq v.
    \end{align}
    Note that the above inequality is in fact an equality for $u \geq 0$.

    Finally, let $\Exp{\lambda}$ denote the exponential distribution with parameter $\lambda > 0$, i.e., the distribution on $\R_{> 0}$ such that $\Exp{\lambda}(x) = \lambda e^{-\lambda x}$.
    
    \section{Warm-Up: Ex-Post AboveThreshold Mechanisms}

Before we present our full ex-post DP Hyperparameter Tuning algorithm, it would be helpful to recall the Sparse Vector Technique~\citep{DworkNRRV09}. In particular, our ex-post DP Hyperparameter Tuning algorithm derives inspiration from the so-called AboveThreshold mechanism and its analysis.

\subsection{Classic AboveThreshold Mechanism}

To state the AboveThreshold mechanism, recall that the sensitivity of a function\footnote{We remark that we keep the range of $f$ discrete for simplicity since we will only use the discrete setting for our generalized algorithm; the result here can be easily extended to continuous function $f$ by changing the Geometric distribution to the Exponential distribution.} $f: \cD \to \Z$ is defined as $\Delta(f) := \max_{D \sim D'} |f(D) - f(D')|$, where the maximum is taken over all neighboring input datasets $D \sim D'$. The setting here is that we are given sensitivity-1 functions\footnote{Again, it is simple to extend the analysis for larger sensitivity; we assume that the sensitivity is $1$ since it suffices for our subsequent modifications.} $f_1, \dots, f_d$ and the goal is to output the first index $i$ such that $f_i(D)$ is at least zero.\footnote{We can easily handle non-zero threshold $\tau$ by considering $f_i - \tau$ instead.} The mechanism works by first sampling a Geometric noise $k$ to be its noisy threshold; then, for each $f_i$, we add an independent Geometric noise $y_i$ to it and check if it exceeds the threshold. If it does, we output $i$ and terminate. It turns out that, in addition to $i$, we can get an estimate of $f_i$ (via $f_i(D) + y_i - k$) for free without any additional privacy cost~\citep{DWXWZK23:free-count-release}. A full description is given in \Cref{alg:at-classic}. 

\begin{algorithm}
        \caption{AboveThreshold Mechanism}
        \label{alg:at-classic}
        \begin{algorithmic}
            \PARAMETERS Sensitivity-1 functions $f_i: \cD \to \Z$ and budget parameters $\epsilon_i$ for $i \in [d]$, and additional privacy budget $\eps' > 0$. \\
            \REQUIRE Dataset $D$.
            \STATE Sample $k \sim \Geo{e^{-\epsilon'}}$ \hfill\COMMENT{threshold noise}
            \FOR{$i = 1, \dots, d$}
                \STATE Sample $y_i \sim \Geo{e^{-\epsilon_i}}$ \hfill\COMMENT{query noise}
                \IF{$f_i(D) + y_i \geq k$} 
                \RETURN $(f_i(D) + y_i - k, i)$ and terminate
                \ENDIF
            \ENDFOR
            \RETURN $\perp$
        \end{algorithmic}
    \end{algorithm}

We remark that this algorithm generalizes the standard ex-ante DP AboveThreshold mechanism since we allow the noise for each $f_i$ to have different privacy budget parameter $\eps_i$. Indeed, with this mechanism, we show that the ex-post privacy budget spent for releasing $f_i$ is only $2\eps_i + \eps'$, as stated below. This generalizes the standard guarantee where $\eps_i$'s are all equal.

\begin{theorem}[Ex-post AboveThreshold]\label{theorem:at-classic}
        Define a function $\teps$ such that 
        $\teps(o, i) = 2 \eps_i + \eps'$ for all $o \in \Z_{\geq 0}, i \in [d]$ and $\teps (\perp) = \eps'$.
        Then, \Cref{alg:at-classic} is 
        ex-post $\teps$-DP.
    \end{theorem}

Our proof closely mirrors the ex-ante AboveThreshold DP proof. Namely, for neighboring datasets $D, D'$ and output $(o, i)$, we can couple the Geometric noises such that $k' = k + 1, y'_i = y_i + 1 + f_i(D) - f_i(D')$ and all other noises remain the same. It is not hard to see that, if the algorithm returns $(o, i)$ on $D$, it returns $(o, i)$ on $D'$ as well. Furthermore, due to the property \eqref{eq:coupling-geom} of the Geometric distribution, the probability decreases by at most $e^{2\eps_i + \eps'}$ factor. This idea is formalized below.

\begin{proof}[Proof of \Cref{theorem:at-classic}]
Consider neighboring datasets $\dataset \sim \dataset'$.  
Let $\cA, \cA'$ be the output distributions of \Cref{alg:at-classic} on $\dataset, \dataset'$, respectively. 
Below, we write $\Geo{p}(< x)$ as a shorthand for $\Geo{p}(\{x - 1, x - 2, \dots\}) = \sum_{y = 0}^{x - 1} \Geo{p}(y)$.

        First, consider any output $(o, i)$. This output happens exactly when $y_i = o + k - f_i(D)$ and $y_j < k - f_j(D)$ for all $j < i$.
        Thus, we have
        \begin{align} \label{eq:above-at-classic-expanded}
        \cA(o, i)
        &= \sum_{k=0}^{\infty} \Geo{e^{-\epsilon'}}(k) \cdot \Geo{e^{-\epsilon_{i}}}(o + k - f_i(D)) \prod_{j=1}^{i-1} \Geo{e^{-\epsilon_{j}}}(< k - f_j(D))
        \end{align}

        Since $\Delta(f_i) \leq 1$, we have 
        $o + k - f_i(D) \leq o + k + 1 - f_i(D')$; applying \eqref{eq:coupling-geom} then yields 
        \begin{align} 
        \Geo{e^{-\epsilon_{i}}}(o + k- f_i(D)) &\leq e^{\epsilon_{i} \cdot (1 - f_i(D') + f_i(D))} \cdot \Geo{e^{-\epsilon_{i}}}(o + k + 1 - f_i(D')) \nonumber \\ &\leq e^{2\epsilon_{i}} \cdot \Geo{e^{-\epsilon_{i}}}(o + k + 1 - f_i(D')), \label{eq:released-query-diff}
        \end{align}
        where the second inequality again uses $\Delta(f_i) \leq 1$.

        Furthermore, $\Delta(f_j) \leq 1$ implies $\Geo{e^{-\epsilon_{j}}}(< k - f_j(D)) \leq \Geo{e^{-\epsilon_{j}}}(< k + 1 - f_j(D'))$. Moreover, \eqref{eq:coupling-geom} implies $\Geo{e^{-\epsilon'}}(k) \leq e^{\eps'} \cdot \Geo{e^{-\epsilon'}}(k + 1)$. Plugging these two inequalities and \eqref{eq:released-query-diff} into \eqref{eq:above-at-classic-expanded} yields
        \begin{align*}
        \cA(o, i)
        &\leq \sum_{k=0}^{\infty} e^{\eps'} \cdot \Geo{e^{-\epsilon'}}(k+1) \cdot e^{2\epsilon_{i}} \cdot \Geo{e^{-\epsilon_{i}}}(o + k + 1 - f_i(D')) \cdot \prod_{j=1}^{i-1} \Geo{e^{-\epsilon_{j}}}(< k + 1 - f_j(D')) \\
        &= e^{2\eps_i + \eps'} \sum_{k=0}^{\infty} \Geo{e^{-\epsilon'}}(k+1) \cdot \Geo{e^{-\epsilon_{i}}}(o + k + 1 - f_i(D')) \cdot \prod_{j=1}^{i-1} \Geo{e^{-\epsilon_{j}}}(< k + 1 - f_j(D')) \\
        &\leq e^{2\eps_i + \eps'} \cdot \cA'(o, i).
        \end{align*}

        Next, consider the output $\perp$. This output happens when $y_j < k - f_j(D)$ for all $j \in [d]$. Thus, we have
        \begin{align*}
        \cA(\perp) &~=~ \sum_{k=0}^{\infty} \Geo{e^{-\epsilon'}}(k) \cdot \prod_{j \in [d]} \Geo{e^{-\epsilon_{j}}}(< k - f_j(D)) \\
        &~\leq~ \sum_{k=0}^{\infty} \Geo{e^{-\epsilon'}}(k) \cdot \prod_{j \in [d]} \Geo{e^{-\epsilon_{j}}}(< k + 1 - f_j(D')) \\
        &~\leq~ \sum_{k=0}^{\infty} e^{\eps'} \cdot \Geo{e^{-\epsilon'}}(k + 1) \cdot \prod_{j \in [d]} \Geo{e^{-\epsilon_{j}}}(< k + 1 - f_j(D')) \\
        &~=~ e^{\eps'}  \sum_{k=0}^{\infty} \Geo{e^{-\epsilon'}}(k + 1) \cdot \prod_{j \in [d]} \Geo{e^{-\epsilon_{j}}}(< k + 1 - f_j(D')) \\
        &~\leq~ e^{\eps'} \cdot \cA'(\perp),
        \end{align*}
        where again we use $\Delta(f) \leq 1$ in the first inequality and \eqref{eq:coupling-geom} in the subsequent inequality.
\end{proof}

\subsubsection{Optimized Noise via Monotonicity}

Let $\succeq$ denote any total order on $\cD$. We say that a function $f$ is monotone (with respect to $\succeq, \sim$) iff the following holds: $f(D) \geq f(D')$ for all $D \sim D'$ such that $D \succeq D'$. An example of this is when $\sim$ denotes an add-remove neighboring notion, i.e., $D \sim D'$ iff $D$ results from adding or removing a user from $D'$; in this case, we may let $\succeq$ be based on the size of the dataset, and $f$ is monotone iff adding a user does not decrease the function value. Such a property holds when $f$ is counting the number of users satisfying certain criteria, which is an example used in our experiment in \Cref{sec:exp}.

For monotone $f$, the same algorithm (\Cref{alg:at-classic}) yields a better ex-post guarantee, where we do not need to pay the factor of $2$ in front of $\eps_i$, as stated below. The proof proceeds similarly to before except that, in the monotone case, either (i) $f_i(D') \geq f_i(D)$ in which case the difference $y'_i - y_i$ is already at most one (instead of two as before), or (ii) $f_i(D') \leq f_i(D)$ in which case we can instead couple with $k' = k, y'_i = y_i + f_i(D) - f_i(D')$ resulting in $y'_i - y_i \leq 1$ again.

\begin{theorem}[Ex-post Monotone AboveThreshold]\label{theorem:at-classic-monotone}
        Define a function $\teps$ such that 
        $\teps(i) = \eps_i + \eps'$ and $\teps (\perp) = \eps'$.
        If $f$ is monotone, then \Cref{alg:at-classic} is 
        ex-post $\teps$-DP.
    \end{theorem}

 \begin{proof}
 We use similar notations as in the proof of \Cref{theorem:at-classic}. The case of $\perp$ output is exactly the same as in that proof. For the output $(o, i)$, we consider the following two subcases, based on whether $D' \succeq D$. First, let us consider the case $D' \succeq D$. In this case, the proof is exactly the same as before except that, since $f_i(D') \geq f_i(D)$, in \eqref{eq:released-query-diff}, we instead get
 \begin{align*} 
        \Geo{e^{-\epsilon_{i}}}(o + k- f_i(D)) &\leq e^{\epsilon_{i} \cdot (1 - f_i(D') + f_i(D))} \cdot \Geo{e^{-\epsilon_{i}}}(o + k + 1 - f_i(D')) \\ &\leq e^{\epsilon_{i}} \cdot \Geo{e^{-\epsilon_{i}}}(o + k + 1 - f_i(D')).
 \end{align*}
 Following the same line of reasoning as before, we then get $\cA(o, i) \leq e^{\eps_i + \eps'} \cdot \cA'(o, i)$ as desired.

 Finally, let us consider the case $D \succeq D'$. In this case, since $f_j(D) \geq f_j(D')$, we have
 \begin{align*}
        \cA(o, i)
        &= \sum_{k=0}^{\infty} \Geo{e^{-\epsilon'}}(k) \cdot \Geo{e^{-\epsilon_{i}}}(o + k - f_i(D)) \prod_{j=1}^{i-1} \Geo{e^{-\epsilon_{j}}}(< k - f_j(D)) \\
        &\leq \sum_{k=0}^{\infty} \Geo{e^{-\epsilon'}}(k) \cdot \Geo{e^{-\epsilon_{i}}}(o + k - f_i(D)) \prod_{j=1}^{i-1} \Geo{e^{-\epsilon_{j}}}(< k - f_j(D'))
 \end{align*}
 Since $o + k - f_i(D) \leq o + k - f_i(D')$, we can apply \eqref{eq:coupling-geom} to arrive at 
        \begin{align*} 
        \Geo{e^{-\epsilon_{i}}}(o + k- f_i(D)) &\leq e^{\epsilon_{i} \cdot (f_i(D) - f_i(D'))} \cdot \Geo{e^{-\epsilon_{i}}}(o + k - f_i(D')) \\ &\leq e^{\epsilon_{i}} \cdot \Geo{e^{-\epsilon_{i}}}(o + k - f_i(D')),
        \end{align*}
        where the second inequality follows from $\Delta(f_i) \leq 1$.

 Combining the above two inequalities then gives
  \begin{align*}
        \cA(o, i)
        &\leq \sum_{k=0}^{\infty} \Geo{e^{-\epsilon'}}(k) \cdot \Geo{e^{-\epsilon_{i}}}(o + k - f_i(D)) \prod_{j=1}^{i-1} \Geo{e^{-\epsilon_{j}}}(< k - f_j(D')) \\
        &\leq \sum_{k=0}^{\infty} \Geo{e^{-\epsilon'}}(k) \cdot e^{\epsilon_{i}} \cdot \Geo{e^{-\epsilon_{i}}}(o + k - f_i(D')) \prod_{j=1}^{i-1} \Geo{e^{-\epsilon_{j}}}(< k - f_j(D')) \\
        &= e^{\eps_i} \cdot \cA'(o, i) \\
        &\leq e^{\eps_i + \eps'} \cdot \cA'(o, i),
 \end{align*}
 which concludes our proof.
 \end{proof}

 \subsection{Generalized AboveThreshold Mechanism via Random Dropping}

In this subsection, we present a new generalized algorithm for AboveThreshold with ex-post DP guarantees (\Cref{alg:generalized-at}). We consider the following general setting: We have mechanisms $M_1, \dots, M_d: \cD \to \cO$ and the goal is to output the first mechanism such that $M_i(D)$ is at least a certain threshold $\tau \in \cO$. Our main idea is \emph{random dropping}, where, instead of always comparing $M_i(D)$ with $\tau$, we only compare with a certain probability; otherwise, we drop $M_i(D)$ completely. While this bears some similarity with the \emph{random stopping} technique of \cite{LT19}, our main innovation is the use of \emph{correlated randomness} $k$ that is sampled at the beginning of the algorithm and determines the dropping probabilities of \emph{all} the mechanisms. This idea is inspired by the above analysis of the AboveThreshold mechanism. Indeed, the high-level structure of our proof follows that of AboveThreshold: we couple $k$ with $k + 1$ in the two neighboring datasets, and bound the ratio of the output probabilities in the two cases. This is formalized in the proof below.

 \begin{algorithm}
        \caption{Generalized AboveThreshold Mechanism with Random Dropping.}
        \label{alg:generalized-at}
        \begin{algorithmic}
            \PARAMETERS Distribution $\cE$, $\epsilon_i$-DP mechanisms $\cM_i: \cD \to \cO$, additional privacy budget $\eps' > 0$, and threshold $\tau \in \cO$ \\
            \REQUIRE Dataset $D$.
            \STATE Sample $k \sim \Geo{e^{-\epsilon'}}$
            \FOR{$i = 1, \dots, d$}
                \STATE Sample $y_i \sim \Ber(e^{-\epsilon_i \cdot k})$ \hfill\COMMENT{random drop}
                \IF{$y_i = 1$} 
                    \STATE $o_i \gets \mechanism_i(D_i)$
                    \IF{$o_i \geq \tau$}
                    \RETURN $(o_i, i)$
                    \ENDIF
                \ENDIF
            \ENDFOR
            \RETURN $\perp$
        \end{algorithmic}
    \end{algorithm}

 Our privacy guarantee for \Cref{alg:generalized-at} is stated below. 
    It says that, if the final output is from $\mechanism_i$, then the privacy budget we pay is only $2\eps_i + \eps'$. 
    The value of $\eps'$ can be arbitrarily small, although setting it too small results in a larger drop probability. 
    The latter can be mitigated by repeating each mechanism $\mechanism_i$ multiple times in the input, which allows us to set that the desired expected number of times that each mechanism is run.

 \begin{theorem}[Ex-post Generalized AboveThreshold]\label{theorem:generalized-at}
        Define a function $\teps$ such that 
        $\teps(o, i) = 2 \eps_i + \eps'$ and $\teps (\perp) = \eps'$.
        Then, \Cref{alg:generalized-at} is 
        ex-post $\teps$-DP.
    \end{theorem}

 \begin{proof}
 Consider neighboring datasets $\dataset \sim \dataset'$. 
Let $\cA, \cA'$ be the output distributions of \Cref{alg:maximum-selection} on $\dataset, \dataset'$, 
respectively and let $\cQ_i, \cQ'_i$ be the output distributions of $\mechanism_i$ on $\dataset,\dataset'$, respectively. Furthermore, let $\cO_{\geq \tau} := \{o' \in \cO \mid o' \geq \tau\}$.

Consider any output $(o, i) \in \cO \times [d]$. Note that, if $o < \tau$, then $\cA(o, i) = \cA'(o, i) = 0$. Otherwise, if $o \geq \tau$, then the algorithm outputs $(o, i)$ iff $y_j = 0$ or $o_j < \tau$ for all $j < i$, and $y_i = 1$ and $o_i = o$. Thus, we have
\begin{align*}
\cA(o, i) &= \sum_{k = 0}^\infty 
                    \Geo{e^{-\epsilon'}}(k) \cdot e^{-\epsilon_i k} \cQ_i(o) \cdot
                                \prod_{j < i} \left(1 - e^{-\epsilon_j k} \cQ_j(\cO_{\geq \tau})\right).
\end{align*}
Since $\mechanism_j$ is $\epsilon_j$-DP, it holds that $\cQ_j(\cO_{\geq \tau}) \ge e^{-\epsilon_j} \cQ'_j(\cO_{\geq \tau})$. Similarly, we have $\cQ_i(o) \leq e^{\eps_i} \cdot Q'_i(o)$. Finally, \eqref{eq:coupling-geom} implies that $\Geo{e^{-\epsilon'}}(k) \leq e^{\eps'} \cdot \Geo{e^{-\epsilon'}}(k+1)$. Plugging these into the above gives
\begin{align*}
\cA(o, i) &\leq \sum_{k = 0}^\infty 
                    \left(e^{\eps'} \cdot \Geo{e^{-\epsilon'}}(k+1)\right) \cdot e^{-\epsilon_i k} \cdot \left(e^{\eps_i} \cQ'_i(o)\right) \cdot
                                \prod_{j \neq i} \left(1 - e^{-\epsilon_j k} \cdot \left(e^{-\eps_j} \cQ'_j(\cO_{\geq \tau})\right)\right) \\
        &= e^{2\eps_i + \eps'} \sum_{k = 0}^\infty 
                    \Geo{e^{-\epsilon'}}(k + 1) \cdot e^{-\epsilon_i(k+1)} \cQ'_i(o) \cdot
                                \prod_{j \neq i} \left(1 - e^{-\epsilon_j (k + 1)} \cQ'_j(\cO_{\geq \tau})\right) \\
        &\leq e^{2\eps_i + \eps'} \cdot \cA'(o, i).
\end{align*}
Finally, consider the output $\perp$. For the algorithm to output $\perp$, we must have $y_j = 0$ or $o_j < \tau$ for all $j \in [d]$. Similar to above, we thus have
\begin{align*}
\cA(\perp) &= \sum_{k = 0}^\infty 
                    \Geo{e^{-\epsilon'}}(k) \cdot \prod_{j \in [d]} \left(1 - e^{-\epsilon_j k} \cQ_j(\cO_{\geq \tau})\right). \\
           &\leq \sum_{k = 0}^\infty 
                    \left(e^{\eps'} \cdot \Geo{e^{-\epsilon'}}(k+1)\right) \cdot \prod_{j \in [d]} \left(1 - e^{-\epsilon_j k} \cdot \left(e^{-\eps_j} \cQ'_j(\cO_{\geq \tau})\right)\right). \\
          &= e^{\eps'} \sum_{k = 0}^\infty 
                    \Geo{e^{-\epsilon'}}(k + 1) \cdot
                                \prod_{j \in [d]} \left(1 - e^{-\epsilon_j (k + 1)} \cQ'_j(\cO_{\geq \tau})\right) \\
          &\leq e^{\eps'} \cdot \cA'(\perp). &\qedhere
\end{align*}
 \end{proof}

 \Cref{theorem:generalized-at} can be viewed as a generalization of \cite{LT19} who prove a similar statement for the case of ex-ante DP (with $\eps_i$'s being all equal). Nevertheless, we stress that our mechanism is based on a different technique. As demonstrated in the next section, our technique is more robust as it generalizes to hyperparameter tuning (without a known threshold) with a similar privacy guarantee, whereas \cite{LT19} have to pay a factor of 3 instead of 2 in that setting.

    \section{Ex-Post Hyperparameter Tuning}


    While \Cref{alg:generalized-at} can already be used for hyperparameter tuning, it requires us to have a good threshold $\tau$ (e.g., a desired accuracy guarantee). However, such a good threshold is hard to determine without a significant prior knowledge. In this section, we give an algorithm that works without such a prior knowledge. Our algorithm is based on random dropping with correlated randomness once again, but now we simply construct a candidate set based on all the outputs that are not dropped and then select the maximum among these candidates. \Cref{alg:maximum-selection} contains the full description.


    For convenience, we extend the order on $\cO$ to $(\cO \times [d]) \cup \{\perp\}$, where $\perp$ is the minimum element, and elements in $\cO \times [d]$ are ordered lexicographically.

     \begin{algorithm}
        \caption{Hyperparameter Tuning Mechanism with Random Dropping}
        \label{alg:maximum-selection}
        \begin{algorithmic}
            \PARAMETERS Distribution $\cE$, Mechanisms $\cM_i: \cD \to \cO$ and budget parameters $\eps_i$ for $i \in [d]$\\
            \REQUIRE Dataset $D$.
            \STATE $S \gets \{\perp\}$
            \STATE Sample $k \sim \cE$
            \FOR{$i = 1, \dots, d$}
                \STATE Sample $y_i \sim \Ber(e^{-\eps_i \cdot k})$ \hfill\COMMENT{random drop}
                \IF{$y_i = 1$} 
                    \STATE $o \gets \mechanism_i(D_i)$
                    \STATE $S \gets S \cup \{(o, i)\}$
                \ENDIF
            \ENDFOR
            \RETURN maximum element in $S$\hfill\COMMENT{as per the total order on $(\cO \times [d]) \cup \{\perp\}$}
        \end{algorithmic}
    \end{algorithm}

    \subsection{Pure-DP}

    Our ex-post pure-DP guarantee for \Cref{alg:maximum-selection} is similar to that of \Cref{alg:generalized-at}, with a small improvement: Here we can take $\teps(\perp) = 0$ instead of $\teps(\perp) = \eps'$ as in \Cref{theorem:generalized-at}.

    \begin{theorem}[Ex-post Pure-DP]\label{theorem:max-score}
        Let $\eps' > 0$ and let each $\mechanism_i$ be
        $\eps_i$-DP.
        Define a function $\teps$ such that 
        $\teps(o, i) = 2 \eps_i + \eps'$ and $\teps (\perp) = 0$.
        Then, \Cref{alg:maximum-selection} with $\cE = \Geo{e^{-\eps'}}$ is 
        ex-post $\teps$-DP.
    \end{theorem}
    \begin{proof}
        Consider neighboring datasets $\dataset \sim \dataset'$. 
        Let $\cA, \cA'$ be the output distributions of \Cref{alg:maximum-selection} on $\dataset, \dataset'$, 
        respectively and let $\cQ_i, \cQ'_i$ be the output distributions of $\mechanism_i$ on $\dataset,\dataset'$, respectively.  
    
        First, the probability that  \Cref{alg:maximum-selection} outputs $\bot$ is independent of 
        input dataset and so $\cA(\perp) = \cA'(\perp)$.  
        Next, consider any output $(o, i) \in \cO \times [d]$. For each $j \in [d] \setminus \{i\}$, let $U^j_{o, i} := \{o' \in \cO \mid (o', j) > (o, i)\}$.
        Since $\mechanism_j$ is $\eps_j$-DP, it holds that 
        $\cQ_j(U^j_{o, i}) \ge e^{-\eps_j} \cQ'_j(U^j_{o, i})$.
        Similarly, we have $\cQ_i(o) \leq e^{\eps_i} \cQ'_i(o)$. Finally, \eqref{eq:coupling-geom} yields $\Geo{e^{-\eps'}}(k) \leq e^{\eps'} \cdot \Geo{e^{-\eps'}}(k+1)$. Thus, we have
        \begin{align*}
        \cA(o, i) &= \sum_{k = 0}^\infty 
                    \Geo{e^{-\eps'}}(k) \cdot e^{-\eps_i k} \cQ_i(o) \cdot
                                \prod_{j \neq i} \left(1 - e^{-\eps_j k} \cQ_j(U^j_{o, i})\right) \\
        &\leq \sum_{k = 0}^\infty 
                    \left(e^{\eps'} \cdot \Geo{e^{-\eps'}}(k+1)\right) \cdot e^{-\eps_i k} \cdot \left(e^{\eps_i} \cQ'_i(o)\right) \cdot
                                \prod_{j \neq i} \left(1 - e^{-\eps_j k} \cdot \left(e^{-\eps_j} \cQ'_j(U^j_{o, i})\right)\right) \\
        &= e^{2\eps_i + \eps'} \sum_{k = 0}^\infty 
                    \Geo{e^{-\eps'}}(k + 1) \cdot e^{-\eps_i(k+1)} \cQ'_i(o) \cdot
                                \prod_{j \neq i} \left(1 - e^{-\eps_j (k + 1)} \cQ'_j(U^j_{o, i})\right) \\
        &\leq e^{2\eps_i + \eps'} \cdot \cA'(o, i) \qedhere
        \end{align*}
    \end{proof}

    The state-of-the-art (ex-ante) pure-DP hyperparameter tuning from~\citep[Corollary 3]{PT22} can only take in a single mechanism $Q$ that is $\eps$-DP.
    To compare this with our mechanism, 
    consider the case where $\cM_1 = \cdots = \cM_d = Q$. In this setting, the two mechanisms are 
     equivalent up to the difference in 
    the distribution of the number of times $Q$ is executed. \Cref{fig:stddev-vs-exp} compares the standard deviation versus the mean of these two distributions. While our distribution has a larger variance, we emphasize its several advantages: our proof is completely elementary and our algorithm is more general as it works for different mechanisms with different $\eps_i$'s values.

    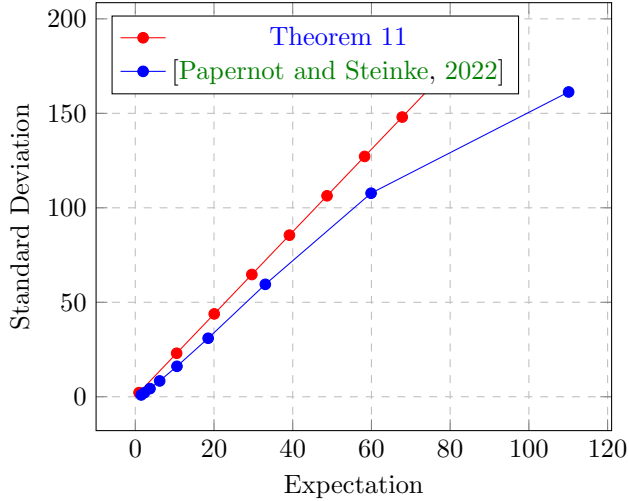
\begin{figure}
        \centering
        \begin{tikzpicture}[scale=1]
    \begin{axis}[
        xlabel={Expectation},
        ylabel={Standard Deviation},
        grid style={dashed},
        grid=major,
        legend pos=north west,
    ]
    
    \addplot[mark=*, color=red] coordinates {
    (0.9552233141976942, 2.1841580037867145)
    (10.507456456174637, 23.02054373477516)
    (20.05968959815158, 43.854834969917235)
    (29.611922740128524, 64.68905545273826)
    (39.16415588210547, 85.52325689092415)
    (48.71638902408241, 106.35745047635052)
    (58.26862216605935, 127.19164006790103)
    (67.8208553080363, 148.02582735195384)
    (77.37308845001324, 168.8600131826156)
    (86.92532159199018, 189.69419803876627)
    };
    \addlegendentry{\Cref{theorem:max-score}}
    
    \addplot[mark=*, color=blue] coordinates {
    (1.4932726172912962, 0.9518452535783978)
    (2.3175071876617723, 2.143464798423721)
    (3.7284097714940327, 4.305360099669469)
    (6.194719440190005, 8.368574625386163)
    (10.584062043356596, 16.10832729937952)
    (18.516012455502484, 30.96496613939242)
    (33.03611294612949, 59.48808048346394)
    (59.90825966255753, 107.7233591452453)
    (110.10243283848793, 161.3056382652037)
    };
    \addlegendentry{\citep{PT22}}
    
    \end{axis}
\end{tikzpicture}
        \caption{
            A plot of the standard deviation vs expectation of the number of 
            invocations of a mechanism by the algorithms
            from Corollary 3 in~\citep{PT22} for $\eta = \eps' / \eps$ and 
            \Cref{theorem:max-score} for $\cE = \Geo{e^{-\eps'}}$ for $\eps = 0.1$ 
            and $\eps' = 0.01$.
        }
        \label{fig:stddev-vs-exp}
    \end{figure}

    In addition, such a repetition trick allows us to prove a utility lower bound which achieves a ``boosting'' effect. To set a stage of the formal statement, note that we will give a relatively weak assumption that at least one of the mechanisms $\mechanism_{i^*}$ outputs a ``good'' candidate with a small probability $\alpha$. The theorem below states that, by repeating each mechanism $\mechanism_i$ a certain number of times $T_i$, we can ensure that \Cref{alg:maximum-selection} outputs a ``good'' candidate with probability at least $1 - \beta$ (where $\beta$ is a small number). We formalize this below, where ``good'' candidates are those that are at least $o^*$. Note also that $T_i$ is an upper bound on the number of times the mechanism $\cM_i$ is run.\footnote{In machine learning settings, the score itself is computed as a measure of performance on a {\em test} set, as a proxy for the measure of performance on the {\em population distribution}. When running more mechanisms, one would need a larger test set in order to get good generalization. This is orthogonal to \Cref{theorem:goodness-guarantee}, which in this setting would refer to $o^*$ as the performance on the population distribution.}\pritish{If we are computing $o$ as the performance (e.g.,  accuracy) on a test set, we need a slightly larger test set size to counter balance this larger variance in the number of times the mechanism is run. Not sure if this is worth mentioning...}\pasin{This is a good point. I'm actually not sure how to mention this without diving too deep into how the score is computed etc...}\pritish{I added a footnote. Can you PTAL? Feel free to rephrase.}
 
    \begin{theorem}\label{theorem:goodness-guarantee}
    Let $\alpha, \beta \in (0, 1)$ and let $T_i = \left\lceil \frac{1}{\alpha} \left(\frac{2}{\beta}\right)^{\eps_i/\eps'} \cdot \ln\left(\frac{2}{\beta}\right) \right\rceil$. Consider \Cref{alg:maximum-selection} with $\cE = \Geo{e^{-\eps'}}$ where, for all $i \in [d]$, we repeat $\mechanism_i$ for $T_i$ times in the input parameter sequence.
    If there exists $o^* \in \cO$ and $i^* \in [d]$ such that $\Pr[\mechanism_{i^*}(D) \geq o^*] \geq \alpha$, then \Cref{alg:maximum-selection} outputs an element that is larger than $(o^*, 0)$ with probability at least $1 - \beta$.
    \end{theorem}

    \begin{proof}
    %
     If the final output is \emph{smaller} than $(o^*, i^*)$, then in all runs of $\mechanism_{i^*}$, it either has to be dropped or the output is less than $o^*$ (or both). For a fixed value of $k$, this happens with probability at most $(1 - e^{-\eps_{i^*} \cdot k} \cdot \alpha)^{T_{i^*}} \leq \exp\left(-e^{-\eps_{i^*} \cdot k} \cdot \alpha \cdot T_{i^*}\right)$. Due to our choice of $T_{i^*}$, this is at most $\beta/2$ for $k \leq \ln\left(2/\beta\right) / \eps'$. Thus, the probability that the final output is smaller than $(o^*, i^*)$ is at most $\beta/2 + \Pr\left[k > \ln\left(2/\beta\right) / \eps'\right] \leq \beta$. 
    \end{proof}

    Finally, we remark that, even in the ex-ante setting with all $\eps$'s being equal, \cite{LT19} showed that the additional factor of 2 in the privacy budget is necessary even under a very weak assumption on the utility. This gives a strong evidence that our algorithm (in its generic form) requires such a factor of 2 blow-up as well.
    

    \subsection{R\'{e}nyi DP}

  In this section, we consider the setting where each $\mechanism_i$ satisfies R\'enyi DP (RDP) instead of pure-DP. 
    As alluded to earlier, many popular DP machine learning algorithms, including DP-SGD \citep{AbadiCGMMT016}, do not satisfy pure-DP but are amenable to privacy analysis using RDP. By changing the distribution of $\cE$ from the Geometric distribution (in the pure-DP case)
    to the Exponential distribution, we can show a version of \Cref{theorem:max-score} for RDP.
    \begin{theorem}[Ex-post RDP]
        \label{theorem:max-score-rdp}
        Let $\ell_1, \dots, \ell_d \ge 0$ and let us assume that each $\mechanism_i$ is
        $(\alpha, \eps_i)$-RDP with
        the output set $\resultSet \times \R$. Let $\tau_i$ be the expected 
        number of times $\mechanism_i$ is executed in  \Cref{alg:maximum-selection} 
        \sasha{I added the reference to  \Cref{alg:maximum-selection}, PTAL}
        (note that it is independent of the dataset); and let $\tau = \sum_{i = 1}^d \tau_i$.
        
        Define a function $\teps$ such that 
        $\teps (o, i) = (2 + \ell_i) \eps_i + (1 + \ell_i) \eps' + 
            \frac{\log(\tau + 1) + \sum_{j \neq i} e^{-\eps_j (1 + \alpha\ell_i)}}{\alpha - 1}$ and 
        $\teps(\perp) = \frac{\log(\tau + 1)}{\alpha - 1}$.
        Then, \Cref{alg:maximum-selection} with $\cE = \Exp{\eps'}$ is 
        $(\alpha, \teps)$-RDP.
    \end{theorem}

    Before we prove \Cref{theorem:max-score-rdp}, let us give a high-level overview.
    Recall that in the above proof of \Cref{theorem:max-score} we use the inequality 
    $\left(1 - e^{-\eps_j k} \cQ_j(U^j_{o, i})\right) \leq \left(1 - e^{-\eps_j (k + 1)} \cQ'_j(U^j_{o, i})\right)$ which follows from the assumption that $\mechanism_j$ is $\eps_j$-DP. 
    The main challenge in proving an RDP bound is that such an inequality fails since  the assumption that $\mechanism_j$ is $(\alpha, \eps_j)$-RDP is weaker. 
    To tackle this, we change the coupling: instead of coupling $k$ with $k + 1$, we couple $k$ with $k + 1 + \ell_i$. 
    Note that, since we allow $\ell_i$ to be non-integer (e.g., a value below one), this step necessitates the use of the Exponential distribution instead of the Geometric distribution. 
    This new coupling allows us to instead compare $\left(1 - e^{-\eps_j k} \cQ_j(U^j_{o, i})\right)$ with $\left(1 - e^{-\eps_j (k + 1 + \ell_i)} \cQ'_j(U^j_{o, i})\right)$. 
    Alas, the former is still not necessarily smaller than the latter.  Nevertheless, via a careful argument, we can bound the ratio of these two quantities. 
    Such a bound then ends up as the last term $\frac{\log(\tau + 1) + \sum_{j \neq i} e^{-\eps_j (1 + \alpha\ell_i)}}{\alpha - 1}$ in our RDP guarantee in \Cref{theorem:max-score-rdp}. 

    To prove \Cref{theorem:max-score-rdp}, we start by collecting some useful facts. The first is the following inequality which is sometimes called the ``reverse H\"older's inequality''; we provide the proof for completeness.
    
    \begin{lemma}
    \label{lem:linear-comb-rdp}
        Let $X$ be a random variable and $f, g$ be any functions on $X$. Then, for any $\alpha > 1$, we have
        \[
            \E[f(X)]^{\alpha}\E[g(X)]^{1-\alpha} \leq \E[f(X)^{\alpha}g(X)^{1-\alpha}].
        \]
    \end{lemma}
    \begin{proof}
        By H\"older's inequality, we have
        \[
          \E[f(X)^{\alpha}g(X)^{1-\alpha}]^{\frac{1}{\alpha}} 
          \E[g(X)]^{\frac{\alpha - 1}{\alpha}} \geq \E[f(X)].
        \]
        Rearranging this yields the claimed inequality.
    \end{proof}
    
    \begin{lemma} \label{lem:rdp-to-conv-ratio}
        For all $\eps > 0$, $\alpha > 1$, if $a, b \in (0, 1)$ are such that 
        $a^{1 - \alpha} b^{\alpha} \leq e^{\eps (\alpha - 1)}$, then, for all $\ell > 0$,
        \[
            (1 - a)^\alpha \left(1 - e^{-\eps (1 + \ell)}b\right)^{1 - \alpha} 
            \leq 
            \exp\left(e^{-\eps (1 + \alpha\ell)} 
            \right).
        \]
    \end{lemma}
    \begin{proof}
        From $1 + x \leq e^x$ for all $x \in \R$, the LHS is at most $\exp((\alpha - 1) e^{-\eps (1 + \ell)}b - \alpha a)$. 
        It is thus sufficient to bound $(\alpha - 1) e^{-\eps (1 + \ell)}b - \alpha a$.
    
        To do this, observe that the condition $a^{1 - \alpha} b^{\alpha} \leq e^{\eps (\alpha - 1)}$ implies
        \begin{align} \label{eq:rdp-rearranged}
        a \geq e^{-\eps} \cdot b^{\frac{\alpha}{\alpha - 1}}.
        \end{align}
        Thus, we may bound the desired term as follows.
        \begin{align*}
            &(\alpha - 1) e^{-\eps (1 + \ell)}b - \alpha a \\
            &\overset{\eqref{eq:rdp-rearranged}}{\leq} (\alpha - 1) e^{-\eps (1 + \ell)}b - \alpha \cdot  e^{-\eps} \cdot b^{\frac{\alpha}{\alpha - 1}} \\
            &= e^{-\eps}\left((\alpha - 1)e^{-\eps\ell} - \alpha b^{\frac{1}{\alpha - 1}}\right) b \\
            &= e^{-\eps} \left(\frac{\alpha - 1}{\alpha}\right)^{\alpha - 1} \left(\left((\alpha - 1)e^{-\eps\ell} - \alpha b^{\frac{1}{\alpha - 1}}\right)^1 \left(\frac{\alpha}{\alpha - 1} \cdot b^{\frac{1}{\alpha - 1}}\right)^{\alpha - 1}\right) \\
            &\overset{(\star)}{\leq} e^{-\eps} \left(\frac{\alpha - 1}{\alpha}\right)^{\alpha - 1} \left(\frac{\left((\alpha - 1)e^{-\eps\ell} - \alpha b^{\frac{1}{\alpha - 1}}\right) + (\alpha - 1)\cdot\left(\frac{\alpha}{\alpha - 1} \cdot b^{\frac{1}{\alpha - 1}}\right)}{\alpha}\right)^\alpha \\
            &= e^{-\eps} \left(\frac{\alpha - 1}{\alpha}\right)^{\alpha - 1} \left(\frac{(\alpha - 1)e^{-\eps\ell}}{\alpha}\right)^\alpha \\
            &= e^{-\eps (1 + \alpha\ell)} \left(\frac{\alpha - 1}{\alpha}\right)^{2\alpha - 1} \\
            &\leq e^{-\eps (1 + \alpha\ell)},
        \end{align*}
        where we use the weighted AM--GM inequality for $(\star)$.
    \end{proof}
    
    \begin{proof}[Proof of \Cref{theorem:max-score-rdp}]
        We will use the same notations as in the proof of \Cref{theorem:max-score}.
        
        First, let us rearrange the term we wish to bound;
        \begin{align}
            &\sum_{\tilde{o} \in \resultSet \times [d] \cup \{\bot\}} 
                (\cA(\tilde{o}))^\alpha (\cA'(\tilde{o}))^{1 - \alpha} e^{(1 - \alpha)\teps(\tilde{o})} \nonumber \\
            &= (\cA(\perp))^\alpha (\cA'(\perp))^{1 - \alpha} e^{-(\alpha - 1)\teps(\perp)} + 
            \sum_{o \in \resultSet, i \in [d]} 
                (\cA(o, i))^\alpha (\cA'(o, i))^{1 - \alpha} e^{-(\alpha - 1)\teps(o, i)} \nonumber  \\
            &\leq \frac{1}{\tau + 1} + \sum_{o \in \resultSet, i \in [d]} 
                (\cA(o, i))^\alpha (\cA'(o, i))^{1 - \alpha} e^{-(\alpha - 1)\teps(o, i)}. 
            \label{eq:renyi-div-expanded}
        \end{align}
        We will now bound each term $(\cA(o, i))^\alpha (\cA'(o, i))^{1 - \alpha}$ above. 
        We have
        \begin{align*}
            \cA(o, i) &= \E_{x \sim \Exp{\eps'}} \left[
                    e^{-\eps_i \cdot x} \cQ_i(o) 
                        \prod_{j \in [d] \setminus \{i\}} (1 - e^{-\eps_j \cdot x} \cQ_j(U^j_{o, i}))
                \right] \\
            &= \cQ_i(o) \cdot \E_{x \sim \Exp{\eps'}}\left[ 
                e^{-\eps_i \cdot x} \prod_{j \in [d] \setminus \{i\}} (1 - e^{-\eps_j \cdot x} \cQ_j(U^j_{o, i})) \right],
        \end{align*}
        and
        \begin{align*}
        &\cA'(o, i) \\ &= \E_{x \sim \Exp{\eps'}} \left[
                    e^{-\eps_i \cdot x} \cQ'_i(o) 
                        \prod_{j \in [d] \setminus \{i\}} (1 - e^{-\eps_j \cdot x} \cQ'_j(U^j_{o, i}))
                \right] \\
        &= \int_0^{\infty} \eps' e^{-\eps' x} \cdot \left(
                e^{-\eps_i \cdot x} \cQ'_i(o)
            \right) \prod_{j \in [d] \setminus \{i\}} (1 - e^{-\eps_j \cdot x} \cQ'_j(U^j_{o, i})) \,dx \\
        &\geq \int_{1+\ell_i}^{\infty} \eps' e^{-\eps' x} \cdot \left(
                e^{-\eps_i \cdot x} \cQ'_i(o)
            \right) \prod_{j \in [d] \setminus \{i\}} (1 - e^{-\eps_j \cdot x} \cQ'_j(U^j_{o, i})) \,dx \\
        &= e^{-(1 + \ell_i)(\eps' + \eps_i)} \cQ'_i(o) \cdot 
            \int_{0}^{\infty} \eps' e^{-\eps' x} \cdot e^{-\eps_i \cdot x} 
                \prod_{j \in [d] \setminus \{i\}} (1 - e^{-\eps_j(1 + \ell_i)} \cdot e^{-\eps_j \cdot x} \cQ'_j(U^j_{o, i})) \,dx \\
        &= e^{-(1 + \ell_i)(\eps' + \eps_i)} \cQ'_i(o) \cdot 
            \E_{x \sim \Exp{\eps'}}\left[
                e^{-\eps_i \cdot x} 
                \prod_{j \in [d] \setminus \{i\}} (1 - e^{-\eps_j(1 + \ell_i)} \cdot e^{-\eps_j \cdot x} \cQ'_j(U^j_{o, i}))
            \right],
        \end{align*}
        where the integrals are due to exponential random variables $x$.
    
        Combining these two inequalities together with \Cref{lem:linear-comb-rdp}, we get that
        \begin{align*}
            &(\cA(o, i))^{\alpha} (\cA'(o, i))^{1 - \alpha} \\
            &\leq e^{(\alpha - 1)(1 + \ell_i)(\eps' + \eps_i)} (\cQ_i(o))^{\alpha} (\cQ_i'(o))^{1 - \alpha} \\
            &\qquad \cdot \E_{x \sim \Exp{\eps'}}\left[
                e^{-\eps_i \cdot x} \cdot 
                \prod_{j \in [d] \setminus \{i\}}
                    (1 - e^{-\eps_i \cdot x} \cQ_j(U^j_{o, i}))^{\alpha} 
                    (1 - e^{-\eps_i(1 + \ell_i)} \cdot e^{-\eps_i \cdot x} \cQ'_j(U^j_{o, i}))^{1 - \alpha} \right].
        \end{align*}
        To bound the inner term, first consider a post-processing of mechanism $M_j$ where, 
        after running, we only output whether the score is greater than $s_i$. 
        Since this is a post-processing of $M_j$, this mechanism is also $(\alpha, \eps_j)$-RDP. 
        As such, we have $(\cQ_j(U^j_{o, i}))^{1 - \alpha} (\cQ'_j(U^j_{o, i}))^{\alpha} \leq e^{\eps_j (\alpha - 1)}$. 
        Thus, we may apply \Cref{lem:rdp-to-conv-ratio} to conclude that  
        \begin{align*}
            (1 - e^{-\eps_i \cdot x} \cQ_j(U^j_{o, i}))^{\alpha} 
                (1 - e^{-\eps_i(1 + \ell_i)} \cdot e^{-\eps_i \cdot x} \cQ'_j(U^j_{o, i}))^{1 - \alpha}
            \leq \exp\left(e^{-\eps_j(1 + \alpha\ell_i)}
            \right).
        \end{align*}
        Plugging this into the above, we arrive at
        \begin{align*}
            &(\cA(o, i))^{\alpha} (\cA'(o, i))^{1 - \alpha} \\
            &\leq e^{(\alpha - 1)(1 + \ell)(\eps' + \eps_i)}
                (\cQ_i(o))^{\alpha} (\cQ'_i(o))^{1 - \alpha} \cdot 
                    \exp\left(\sum_{j \in [d] \setminus \{i\}} e^{-\eps_j(1 + \alpha\ell_i)}\right) 
                    \E_{x \sim \Exp{\eps'}}\left[e^{-\eps_i \cdot x}\right]\\
            &= e^{(\alpha - 1)(1 + \ell)(\eps' + \eps_i)}
                    (\cQ_i(o))^{\alpha} (\cQ'_i(o))^{1 - \alpha} \cdot 
                        \exp\left(\sum_{j \in [d] \setminus \{i\}} e^{-\eps_j(1 + \alpha\ell_i)}\right) 
                        \tau_i \\
            &\leq \frac{\tau_i}{\tau + 1} \cdot
                e^{(\alpha - 1)(\teps(o, i) - \eps_i)} (\cQ_i(o))^{\alpha} (\cQ'_i(o))^{1 - \alpha},
        \end{align*}
        where in the last inequality we use our choice of $\ell_i$ and $\teps(o, i)$.
    
        Combining with \eqref{eq:renyi-div-expanded}, we thus get
        \begin{align*}
            &\sum_{\tilde{o} \in \resultSet \times [d] \cup \{\bot\}} (\cA(\tilde{o}))^{\alpha} (\cA'(\tilde{o}))^{1 - \alpha} e^{(1 - \alpha)\teps(\tilde{o})} \\ 
            &\leq \frac{1}{\tau + 1} + 
                \sum_{i \in [d]} \,
                    \sum_{o \in \resultSet} 
                        \frac{\tau_i}{\tau + 1} e^{-(\alpha - 1) \eps_i} 
                            \frac{(\cQ_i(o))^{\alpha}}{(\cQ'_i(o))^{\alpha - 1}} \\
            &= \frac{1}{\tau + 1} + 
                \sum_{i \in [d]} \,
                    \frac{\tau_i}{\tau + 1} e^{-(\alpha - 1) \eps_i} 
                    \left(
                        \sum_{o \in \resultSet} 
                            \frac{(\cQ_i(o))^{\alpha}}{(\cQ'_i(o))^{\alpha - 1}} 
                    \right) \\
            &\le \frac{1}{\tau + 1} + \sum_{i \in [d]} \frac{\tau_i}{\tau + 1} \cdot e^{-(\alpha - 1)\eps_i} e^{(\alpha - 1) \eps_i} \\
            &= \frac{1}{\tau + 1} + \sum_{i \in [d]} \frac{\tau_i}{\tau + 1} = 1.
            \qedhere
        \end{align*}
    \end{proof}

    We note that, unlike \Cref{theorem:max-score}, $\teps(\perp) \ne 0$ in \Cref{theorem:max-score-rdp}, i.e., we pay a privacy budget even when we fail to output anything meaningful.
    Again, this can be mitigated by repeating each mechanism multiple times in the input to decrease the probability of outputting $\perp$ to be arbitrarily small.

    When the $\eps_i$'s are different, it might be beneficial to pick $\ell_i$'s to be different as well. 
    On the other hand, if we only consider the simple setting when $\eps_1 = \dots = \eps_d = \eps$ and we wish to choose $\ell_1, \dots, \ell_d$ to all be equal to $\ell$. 
    Then, it is not hard to verify that by setting $\ell = O\left(\frac{\log d}{\eps \alpha}\right)$, we can ensure that $\sum_{j \in [d]} e^{-\eps_j (1 + \alpha\ell)} \leq 1$. 
    With this setting of parameters and assuming $\eps' \leq O(\eps)$, we thus have the RDP bound of $\teps = 2\eps + \eps' + O\left(\frac{\log d}{\alpha}\right)$. 
    Note that this is similar to the bound from state-of-the-art (ex-ante) RDP hyperparameter tuning from \citep[Theorem 2]{PT22}, 
    which gives an RDP bound of $(2 + \eta)\eps + O\left(\frac{\log d}{\lambda}\right)$, where $\eta$ is the parameter of the negative binomial distribution (and assuming $\gamma \in (0, 1)$ is a constant and $\hat{\lambda} = \lambda, \hat{\eps} = \eps$).

    Alternatively, one may notice that $\teps (o, i)$ doesn't depend on $\ell_j$ for $j \neq i$;
    hence, for each $i$ it is possible to choose $\ell_i$ as a value minimizing $\teps (o, i)$.
    
    \section{Fully-Adaptive Composition with Ex-Post R\'enyi DP}

    Real-life applications of DP mechanism are often highly interactive: 
    i.e., the analyst queries private data and based on the results of these queries decides what to query next. 
    Moreover, often it is important to be able to choose further privacy parameters based on previous responses. 
    Following \cite{RVRU16}, we express this interactivity in a form of a ``game'' between an adversary $\adversary$
    and some system $\mathcal{F}_{\alpha, \eps}$.
    In this interaction there is an unknown bit that the adversary wishes to learn; 
    on each step $i$ the adversary (based on previous responses)
    chooses two datasets $\dataset^{(0)}_i$ and $\dataset^{(1)}_i$, a privacy loss function $\teps_i$, and a mechanism
    $\mechanism_i$ that is $(\alpha, \teps_i)$-RDP; the system
    decides if such request could be answered;
    and if the system allows to proceed, the result $\mechanism_i(D^{(b)}_i)$ is given to the adversary. 
    The \emph{privacy filter} we devise is simple: Start with a total RDP budget $\eps$, 
    subtract from it the ex-post RDP bound after each request is answered, and only allow the next request to be answered
    if the remaining budget is at least the maximum possible ex-post RDP bound of the mechanism. 
    See \Cref{alg:rdp-filter} for the details.

    \begin{algorithm}
        \caption{Privacy filter for ex-post RDP.}
        \label{alg:rdp-filter}
        \begin{algorithmic}
            \PARAMETERS Order $\alpha > 1$,
            privacy budget $\eps > 0$,
            number of steps $n$.
            \REQUIRE Adversary $\adversary$,
            private bit $b \in \{0, 1\}$.

            \FOR{$i$ from $1$ to $n$}

                \STATE 
                    $\dataset^{(0)}_i, \dataset^{(1)}_i, 
                        \teps_i, \mechanism_i \gets \adversary(o_1, \dots, o_{i - 1})$

                \IF{$\sum_{j = 1}^{i - 1} 
                        \teps_j(o_j) + \sup_o \teps_i(o) > \eps$}
                    \RETURN $o_1$, \dots, $o_{i - 1}$
                \ELSE
                    \STATE $o_i \gets \mechanism_i(\dataset^{(b)}_i)$
                \ENDIF
            \ENDFOR
            \RETURN $o_1$, \dots, $o_n$
        \end{algorithmic}
    \end{algorithm}

    Our privacy filter allows us to use ex-post RDP algorithms in interactive manners while ensuring a final ex-ante RDP bound. 
    This result extends the results of \cite{Lcu21,FZ21} to allow adversary to issue mechanisms with ex-post guarantees.
    We note that such a connection between ex-post DP and ex-ante DP via privacy filter has been made before, e.g., \citep{LebensoldPB24, RVRU16},
    but only for pure-DP and approximate-DP.  
    To the best of our knowledge, this work is the first to generalize this result to RDP.

    \begin{theorem}
        \label{theorem:rdp-filter}
        For any adversary $\adversary$, 
        $\alpha > 1$, 
        $\eps > 0$, $n \in \N$, 
        \[
            \renyidiv{\alpha}{
                \IT{0}{\eps}{\alpha}{\adversary}
            }{
                \IT{1}{\eps}{\alpha}{\adversary}
            } \le \eps,
        \]
        where $\IT{b}{\eps}{\alpha}{\adversary}$ is the output of \Cref{alg:rdp-filter}.
    \end{theorem}
    \begin{proof}
        Without loss of generality, we can assume that the adversary is always issuing queries 
        such that $\sum_{j = 1}^{i - 1} \teps_{j}(o_j) + \sup_o \teps_{i}(o_i) \le \eps$
        for all $i \in [n]$.
        Let us denote the query issued by the adversary after seeing $o_1$,\dots, $o_{i - 1}$ as
        $D^{(0)}_{o_1, \dots, o_{i - 1}}$, $D^{(1)}_{o_1, \dots, o_{i - 1}}$, 
        $\teps_{o_1, \dots, o_{i - 1}}$, and 
        $\mechanism_{o_1, \dots, o_{i - 1}}$.
        Let us also denote the distribution of $\mechanism_{o_1, \dots, o_{i - 1}}(D^{b}_{o_1, \dots, o_{i - 1}})$
        by $P^{(b)}_{o_1, \dots, o_{i - 1}}$.
        Note that 
        \begin{align*}
            &e^{(\alpha - 1) 
                \renyidiv{\alpha}{
                    \IT{0}{\eps}{\alpha}{\adversary}
                }{
                    \IT{1}{\eps}{\alpha}{\adversary}
                }
            } \\
            &=
            \sum_{o_1, o_2, \dots, o_n} 
                \frac{
                    (
                        P^{(0)}(o_1) 
                        P^{(0)}_{o_1}(o_2)
                        \cdots
                        P^{(0)}_{o_1, \dots, o_{n - 1}}(o_n)
                    )^{\alpha}
                }{
                    (
                        P^{(1)}(o_1) 
                        P^{(1)}_{o_1}(o_2)
                        \cdots
                        P^{(1)}_{o_1, \dots, o_{n - 1}}(o_n)
                    )^{\alpha - 1}
                } \\
            &= 
            \sum_{o_1, o_2, \dots, o_{n - 1}} 
                \frac{
                    (
                        P^{(0)}(o_1) 
                        P^{(0)}_{o_1}(o_2)
                        \cdots
                        P^{(0)}_{o_1, \dots, o_{n - 2}}(o_{n - 1})
                    )^{\alpha}
                }{
                    (
                        P^{(1)}(o_1) 
                        P^{(1)}_{o_1}(o_2)
                        \cdots
                        P^{(1)}_{o_1, \dots, o_{n - 2}}(o_{n - 1})
                    )^{\alpha - 1}
                }
                \left(
                    \sum_{o_n}
                    \frac{
                        (
                            P^{(0)}_{o_1, \dots, o_{n - 1}}(o_n)
                        )^{\alpha}
                    }{
                        (
                            P^{(1)}_{o_1, \dots, o_{n - 1}}(o_n)
                        )^{\alpha - 1}
                    }
                \right).
        \end{align*}
        Further,  note that $\mechanism_{o_1, \dots, o_{n - 1}}$ is 
        $(\alpha, \teps_{o_1, \dots, o_{n - 1}})$-RDP and hence, 
        \begin{align*}
                \sum_{o_n} 
                \frac{
                    (
                        P^{(0)}_{o_1, \dots, o_{n - 1}}(o_n)
                    )^{\alpha}
                }{
                    (
                        e^{\teps_{o_1, \dots, o_{n - 1}}(o_n)}P^{(1)}_{o_1, \dots, o_{n - 1}}(o_n)
                    )^{\alpha - 1}
                } \le 1.
        \end{align*}
    
        Let $L(o_1, \dots, o_k) = \frac{
                    (
                        P^{(0)}(o_1) 
                        P^{(0)}_{o_1}(o_2)
                        \cdots
                        P^{(0)}_{o_1, \dots, o_{k - 1}}(o_k)
                    )^{\alpha}
                }{
                    (
                        e^{\teps(o_1)}
                        P^{(1)}(o_1) 
                        e^{\teps_{o_1}(o_2)}
                        P^{(1)}_{o_1}(o_2)
                        \cdots
                        e^{\teps_{o_1, \dots, o_{k - 1}}(o_k)}
                        P^{(1)}_{o_1, \dots, o_{k - 1}}(o_k)
                    )^{\alpha - 1}
                }.$
        Then
        \begin{align*}
            & \frac{
                e^{(\alpha - 1) 
                    \renyidiv{\alpha}{
                        \IT{0}{\eps}{\alpha}{\adversary}
                    }{
                        \IT{1}{\eps}{\alpha}{\adversary}
                    }
                }
            }{
                e^{(\alpha - 1) \eps}
            } \\
            & = 
            \frac{
                1
            }{
                e^{(\alpha - 1) \eps}
            } 
            \sum_{o_1, o_2, \dots, o_n}
                \frac{
                    (
                        P^{(0)}(o_1) 
                        P^{(0)}_{o_1}(o_2)
                        \cdots
                        P^{(0)}_{o_1, \dots, o_{n - 1}}(o_n)
                    )^{\alpha}
                }{
                    (
                        P^{(1)}(o_1) 
                        P^{(1)}_{o_1}(o_2)
                        \cdots
                        P^{(1)}_{o_1, \dots, o_{n - 1}}(o_n)
                    )^{\alpha - 1}
                } \\
            & \le
            \sum_{o_1, o_2, \dots, o_n}
                \frac{
                    (
                        P^{(0)}(o_1) 
                        P^{(0)}_{o_1}(o_2)
                        \dots
                        P^{(0)}_{o_1, \cdots, o_{n - 1}}(o_n)
                    )^{\alpha}
                }{
                    (
                        e^{\teps(o_1)}
                        P^{(1)}(o_1) 
                        e^{\teps_{o_1}(o_2)}
                        P^{(1)}_{o_1}(o_2)
                        \cdots
                        e^{\teps_{o_1, \dots, o_{n - 1}}(o_n)}
                        P^{(1)}_{o_1, \dots, o_{n - 1}}(o_n)
                    )^{\alpha - 1}
                } \\
            & =
            \sum_{o_1, \dots, o_n} L(o_1, \dots, o_n) \\
            & =
            \sum_{o_1, \dots, o_{n - 1}}
                L(o_1, \dots, o_{n - 1})
                \left(
                    \sum_{o_n} 
                    \frac{
                        (
                            P^{(0)}_{o_1, \dots, o_{n - 1}}(o_n)
                        )^{\alpha}
                    }{
                        (
                            e^{\teps_{o_1, \dots, o_{n - 1}}(o_n)}P^{(1)}_{o_1, \dots, o_{n - 1}}(o_n)
                        )^{\alpha - 1}
                    }
                \right)
                \\
            & \le
            \sum_{o_1, \dots, o_{n - 1}}
                L(o_1, \dots, o_{n - 1}) \\
            & \le
            \sum_{o_1, \dots, o_{n - 2}}
                L(o_1, \dots, o_{n - 2}) 
                \left(
                    \sum_{o_{n - 1}} 
                    \frac{
                        (
                            P^{(0)}_{o_1, \dots, o_{n - 2}}(o_{n - 1})
                        )^{\alpha}
                    }{
                        (
                            e^{\teps_{o_1, \dots, o_{n - 2}}(o_{n - 1})}P^{(1)}_{o_1, \dots, o_{n - 2}}(o_{n - 1})
                        )^{\alpha - 1}
                    }
                \right) \\
            &\vdots \\
            & \le \sum_{o_1} L(o_1) \le 1, \mbox{which implies that 
        $\renyidiv{\alpha}{
            \IT{0}{\eps}{\alpha}{\adversary}
        }{
            \IT{1}{\eps}{\alpha}{\adversary}
        } \le \eps$.
        } 
        \qedhere
        \end{align*}
    \end{proof}

    \section{Experiments}
    \label{sec:exp}

    We present two sets of experiments: In the first, we evaluate the performance of our algorithm on analytical tasks and in the second, we focus on the performance on a machine learning problem. 

    \subsection{Analytical Problem}

    Informally the problem is as follows~\citep{RSWR23}: given a message board, the goal is to estimate the number of unique users per thread, each with relative error 10\%; however we want as many estimates as possible.
    Here, a user could contribute to any of the threads.
    We consider two datasets.
    \begin{description}[leftmargin=4mm]
        \item[Synthetic:] The synthetic datasets are generated as follows: $N \in \{8000, 16000, 32000, 64000, 128000\}$
        samples are obtained from the power-law distribution with support on $[300]$ 
        (i.e., the distribution such that for $x \in [300]$,
        the density is proportional to $x^{0.75}$ and is $0$ otherwise). We assume that each $x$ corresponds
        to a thread and the number of samples with this value is the number of users.  Hence, we convert these samples 
        into a histogram of $300$ values with their counts.
        \item[Reddit:] We use the {\tt webis/tldr-17} dataset~\citep{volske-etal-2017-tl} that contains authors of posts and subreddits
            where the post was posted. The histogram consists of subreddits (i.e., threads) and the number of unique users who posted in the subreddit.
    \end{description}

    We consider two types of algorithms: one where a pure-DP guarantee is available and another
    where we eventually have an approximate-DP guarantee. 
    However, in both cases, we can check whether the current estimate $\hat{y}$ is good 
    (i.e., we expect it to be with less than 10\% error) by checking that 
    $|(\hat{y} + \sigma) / (\hat{y} - \sigma)| \in [0.9, 1.1]$ and $|\hat{y}| \ge \sigma$, where $\sigma$
    is the standard deviation of the noise used to obtain the estimate.

    In the case of pure-DP, we follow~\citep{RSWR23} and allow mechanisms to compute each estimate with privacy budget 
    $\eps = 0.001 \cdot (\sqrt{2})^i$ for some $i$, with a total available budget of $10$.
    The comparison includes the doubling mechanism with Laplace noise 
    (the algorithm that attempts one $\eps$ after another and pays for them via
    composition)~\citep{WRLWN19}, 
    noise reduction method with Laplace noise from~\citep{WRLWN19}, and \Cref{alg:maximum-selection} 
    with Laplace mechanism and $\eps' = 0.001$.
    The detailed results can be seen in 
    \Cref{table:synthetic-pure-dp}.
    Note that in terms of number of produced answers, our algorithm outperforms all other solutions
    and in terms of precision (percentage of outputs that were indeed with 10\% relative error) is similar 
    to the doubling estimator and within reasonable bounds.
    
    In the case of approximate-DP, we allow mechanisms to compute each estimate with privacy budget 
    $\eps = 0.001 \cdot (\sqrt{2})^i$ for some $i$, with a total available budget of $(10, 10^{-6})$.
    We compare the doubling mechanism with Gaussian noise and zCDP budgeting~\citep{BunS16}, the
    Brownian Motion algorithm with zCDP budgeting~\citep{WRWR22}, and
    \Cref{alg:maximum-selection} with Gaussian mechanism and RDP budgeting.  The results can be seen in 
    \Cref{table:synthetic-approx-dp}.
    In this case, our algorithm underperforms, which is not too surprising since the Gaussian mechanism with 
    zCDP budgeting is tailored for tasks of this nature.
    \pritish{\Cref{table:synthetic-pure-dp}, \Cref{table:synthetic-approx-dp} are not clear to me.
    If ``Provided Answers'' is $14.77$, does it mean we could only estimate that many threads out of $8000$ threads?}
    \sasha{No, the number of threads is $300$, but the total value is $8000$. 
    So $14.77$ means that we estimated this many threads out of $300$}

    \begin{table}[H]
        \centering
        \caption{
            Comparison between pure-DP mechanisms; 
            the column `Produced Answers'
            contains the average and standard deviation of the number of threads that
            the algorithm was able to estimate before the budget got exhausted and 
            the column `Precision' contains the average and standard deviation
            of the fraction of threads that were estimated with less than 10\% relative error
            among the estimated columns.
            A cell value \ans{$a$}{$b$} means $a$ is the average and $b$ is the standard deviation.
            The dataset name S$N$ means synthetic dataset made of $N$ samples.
        }
        \begin{tabular}{rcccccc}
    \toprule
    \multirow{4}{4em}{Dataset} & \multicolumn{2}{c}{Doubling} &
    \multicolumn{2}{c}{Noise Reduction} & \multicolumn{2}{c}{\Cref{alg:maximum-selection}} \\
    & \multicolumn{2}{c}{Mechanism} & \multicolumn{2}{c}{Mechanism}
    & \multicolumn{2}{c}{w/ Laplace} \\[5pt]
    & \multirow{2}{4em}{Precision} & Produced & \multirow{2}{4em}{Precision}  & Produced & \multirow{2}{4em}{Precision}  & Produced \\ 
    & & Answers & & Answers & & Answers \\
    \midrule
     S8000  & \ans{0.911}{0.07} & \ans{14.77}{0.47}
     & \ans{0.999}{0.02} & \ans{2.22}{1.45} 
     & \ans{0.912}{0.06} & \ans{20.37}{0.52}
     \\
     S16000 & \ans{0.912}{0.06} & \ans{22.47}{0.54} 
     & \ans{0.999}{0.02} & \ans{4.47}{2.22} 
     & \ans{0.911}{0.05} & \ans{30.63}{0.57}
     \\
     S32000 & \ans{0.910}{0.05} & \ans{33.96}{0.56} 
     & \ans{0.998}{0.12}  & \ans{8.12}{3.29}
     & \ans{0.905}{0.04} & \ans{45.74}{0.63}
     \\
     S64000  & \ans{0.909}{0.04} & \ans{50.90}{0.61} 
     & \ans{0.998}{1.72} & \ans{15.14}{5.25}
     & \ans{0.911}{0.03} & \ans{68.39}{0.73} 
     \\
     S128000 & \ans{0.909}{0.03} & \ans{76.09}{0.74} 
     & \ans{0.997}{0.01} & \ans{27.68}{9.15}
     & \ans{0.912}{0.03} & \ans{102.1}{0.88} \\
     Reddit & \ans{0.911}{0.02} & \ans{279.7}{1.10} 
     & \ans{0.992}{0.01} & \ans{207.2}{42.1}
     & \ans{0.922}{0.01} & \ans{327.5}{13.9} \\
    \bottomrule
\end{tabular}
        \label{table:synthetic-pure-dp}
    \end{table}

    \begin{table}
    \centering
    \caption{Comparison between approximate-DP mechanisms on synthetic data; 
        the column `Produced Answers'
        contains the average and standard deviation of the number of threads that
        the algorithm was able to estimate before the budget got exhausted and 
        the column `Precision' contains the average and standard deviation
        of the fraction of threads that were estimated with less than 10\% relative error
        among the estimated columns.
        A cell value \ans{$a$}{$b$} means $a$ is the average and $b$ is the standard deviation.
        The dataset name S$N$ means synthetic dataset made of $N$ samples.
    }
    \begin{tabular}{rcccccc}
        \toprule
        \multirow{4}{4em}{Dataset} & \multicolumn{2}{c}{Brownian Motion} &
        \multicolumn{2}{c}{Doubling} & \multicolumn{2}{c}{\Cref{alg:maximum-selection}} \\
        & \multicolumn{2}{c}{Mechanism} & \multicolumn{2}{c}{Mechanism}
        & \multicolumn{2}{c}{w/ Gaussian} \\[5pt]
        & \multirow{2}{3em}{Precision} & Produced & \multirow{2}{3em}{Precision}  & Produced & \multirow{2}{3em}{Precision}  & Produced \\ 
        & & Answers & & Answers & & Answers \\
        \midrule
         S8000  & \ans{0.974}{0.03} & \ans{28.63}{0.74}
         & \ans{0.969}{0.05} & \ans{17.20}{0.57} 
         & \ans{0.970}{0.09} & \ans{6.694}{0.46}
         \\
         S16000 & \ans{0.973}{0.02} & \ans{50.44}{0.89} 
         & \ans{0.971}{0.03} & \ans{30.75}{0.65} 
         & \ans{0.972}{0.05} & \ans{11.97}{0.33}
         \\
         S32000 & \ans{0.974}{0.02} & \ans{88.58}{1.01} 
         & \ans{0.973}{0.02} & \ans{54.74}{0.80}
         & \ans{0.971}{0.04} & \ans{20.61}{0.51}
         \\
         S64000  & \ans{0.975}{0.01} & \ans{154.8}{1.29} 
         & \ans{0.974}{0.02} & \ans{96.95}{1.01}
         & \ans{0.973}{0.03} & \ans{33.84}{0.63} 
         \\
         S128000 & \ans{0.977}{0.01} & \ans{269.7}{1.50} 
         & \ans{0.980}{0.01} & \ans{173.6}{1.23}
         & \ans{0.970}{0.03} & \ans{52.33}{1.17} \\
        \bottomrule
    \end{tabular}
    \label{table:synthetic-approx-dp}
\end{table}

    \subsection{Machine Learning}
    \label{section:ml}
    We perform a second set of experiments related to a machine learning task. 
    \begin{enumerate}[leftmargin=*]
        \item In the first set of experiments we follow the setup from~\citep{WRLWN19,WRWR22} 
            and train a linear regression model on a dataset of timeseries generated by Twitter usage~\citep{twitter-data}
            (subsampled to 100000 data-points) and search for a model with at most $0.05$ MSE.
            We compare the following two mechanisms. 
            \begin{enumerate}
                \item Brownian motion with the AboveThreshold mechanism using sufficient statistics perturbation~\citep{VS09},
                    a sequence $0.1$, $0.2$, \dots $1$ of values of $\eps$ for Brownian motion, and $0.01$ for
                    AboveThreshold on the MSE of the model.
                \item \Cref{alg:maximum-selection} with the DP-SGD~\citep{AbadiCGMMT016} mechanism, learning linear models 
                    with $\eps' = 0.01$, possible values of $\eps$ in $\{0.1, 0.2, \dots 1\}$, learning rate in $\{0.01, 0.1, 1\}$,
                    epochs in $\{1, 5, 10\}$, batch sizes in $\{32, 64, 128, 256, 512, 1000\}$, and clipping
                    norms in $\{0.1, 1, 10\}$.
                \item Doubling mechanism~\cite{WRLWN19} running DP-SGD tuned according to~\cite{PT22} 
                    with identical hyperparameters to those used by \Cref{alg:maximum-selection}.
            \end{enumerate}
        \item In the second experiment, we train a classifier for the MNIST dataset~\citep{lecun2010mnist} 
            and search for the minimal $\eps$ such that the model has at least $0.6$ accuracy.
            We compare the following two mechanisms. 
            \begin{enumerate}
                \item Brownian motion with the AboveThreshold mechanism using output perturbation~\citep{VS09},
                    a sequence $0.1$, $0.2$, \dots $1$ of values of $\eps$ for Brownian motion, and $0.01$ for
                    AboveThreshold on accuracy of the model.
                \item \Cref{alg:maximum-selection} with DP-SGD mechanisms learning CNN models 
                    (for the architecture see \cite{mnist-example})
                    with $\eps' = 0.01$, possible values of $\eps$ in $\{0.1, 0.2, \dots 1\}$, learning rate in $\{0.01, 0.1, 1\}$,
                    epochs in $\{1, 5, 10\}$, batch sizes in $\{32, 64, 128, 256, 512, 1000\}$, and clipping
                    norms in $\{0.1, 1, 10\}$.
                \item Doubling mechanism~\cite{WRLWN19} running DP-SGD tuned according to~\cite{PT22} 
                    with identical hyperparameters to those used by \Cref{alg:maximum-selection}.
            \end{enumerate}
        \item In the third experiment, we train a classifier for the Gisette~\citep{gisette_170} dataset 
            and search for the minimal $\eps$ such that the model has at least $0.4$ accuracy.
            We compare the following two mechanisms. 
            \begin{enumerate}
                \item Brownian motion with the AboveThreshold mechanism using output perturbation~\citep{VS09},
                    a sequence $0.1$, $0.2$, \dots $1$ of values of $\eps$ for Brownian motion, and $0.01$ for
                    AboveThreshold on accuracy of the model.
                \item \Cref{alg:maximum-selection} with DP-SGD mechanisms learning a linear model
                    with $\eps' = 0.01$, possible values of $\eps$ in $\{0.1, 0.2, \dots 1\}$, learning rate in $\{0.01, 0.1, 1\}$,
                    epochs in $\{1, 5, 10\}$, batch sizes in $\{32, 64, 128, 256, 512, 1000\}$, and clipping
                    norms in $\{0.1, 1, 10\}$.
                \item Doubling mechanism~\cite{WRLWN19} running DP-SGD tuned according to~\cite{PT22} 
                    with identical hyperparameters to those used by \Cref{alg:maximum-selection}.
            \end{enumerate}
    \end{enumerate}
    (In both cases we use Opacus \citep{opacus} for training DP-SGD.)
    
    The results of comparison can be seen in \Cref{table:ml}. 
    Our algorithm significantly outperforms the previous Brownian motion algorithms and doubling mechanism. 
    This can be explained by the fact that DP-SGD vastly outperforms the simpler models in these settings~\citep{YZ00L20} 
    and the fact that doubling requires running tuning which (in order to keep the budget small) needs high $\alpha$.
    Our algorithm also consistently outperforms the doubling mechanism. 
    This superior performance can be attributed to the doubling mechanism's privacy
    loss being approximately two times greater than that of the tuning mechanism 
    which in-turn is about two times greater than the underlying procedure.

    \begin{table}[t]
        \centering
        \caption{
            Comparison of the $\eps$ values used by the Brownian Motion mechanism,
            doubling mechanism, and
            \Cref{alg:maximum-selection} when applied to
            machine learning tasks.
            The numbers represent the average ex-post $(\eps, 10^{-6})$-DP
            guarantees over 100 trials.
        }
        \begin{tabular}{lccc}
            \toprule
            Dataset & Brownian Motion & Doubling Mechanism & \Cref{alg:maximum-selection} \\
            \midrule
            Twitter & 0.77 & 0.55 & 0.28 \\
            MNIST   & 0.62 & 0.38 & 0.32 \\
            Gisette & 0.33 & 0.54 & 0.23 \\
            \bottomrule
        \end{tabular}
        \label{table:ml}
    \end{table}

    \section{Conclusion and Open Problems}

    In this work, we give a simple yet general algorithm for DP hyperparameter tuning that works even for ex-post DP and RDP.  Despite its generality, our experiments show that it achieves significant advantage over previous algorithms for ML applications.
    Two immediate questions remain. 
    First, is it possible to get rid of the $\ell$'s in \Cref{theorem:max-score-rdp}? 
    Second, and somewhat related, is the question of proving a zCDP version of the result, 
    which would improve the analysis in the case of analytics workloads since the Gaussian mechanism is typically used in those cases.

    \newpage
    
    \bibliographystyle{abbrvnat}
    \bibliography{ref}
    
\end{document}